\documentclass{style/vldb_style_sample/vldb}
\usepackage{graphicx}
\usepackage{balance}  

\usepackage{booktabs}
\usepackage{amsmath}
\usepackage{verbatim}
\usepackage{algorithmicx}
\usepackage{algorithm}
\usepackage{algpseudocode}
\usepackage{pgfplots}
\usepackage{pgfplotstable}
\usepackage[center]{subfigure}
\usepackage{url}
\usepackage{calc}
\usetikzlibrary{patterns,shapes,arrows,positioning,fit,decorations.pathreplacing}
\usepgfplotslibrary{groupplots}

\begin{document}

\newtheorem{lemma}{Lemma}
\newtheorem{theorem}{Theorem}
\newtheorem{corollary}{Corollary}
\newtheorem{example}{Example}
\newtheorem{assumption}{Assumption}
\newtheorem{remark}{Remark}

\newdef{definition}{Definition}
\newdef{scenario}{Scenario}


\title{Multiple Query Optimization on the D-Wave 2X Adiabatic Quantum Computer}




%
%
%
%

\numberofauthors{1} 

\author{
%
%
\alignauthor
Immanuel Trummer and Christoph Koch\\
			 \email{\{firstname\}.\{lastname\}@epfl.ch}\\
       \affaddr{\'Ecole Polytechnique F\'ed\'erale de Lausanne}
}


\newlength{\unitCellXdistance}
\setlength{\unitCellXdistance}{1.5cm}
\newlength{\unitCellYdistance}
\setlength{\unitCellYdistance}{2.25cm}
\newlength{\qubitXdistance}
\setlength{\qubitXdistance}{0.75cm}
\newlength{\qubitYdistance}
\setlength{\qubitYdistance}{0.5cm}
\newlength{\cellXoffset}
\newlength{\cellYoffset}
\newlength{\minQubitWidth}
\setlength{\minQubitWidth}{4mm}

\def\positionUnitCell#1#2{
\setlength{\cellXoffset}{#1\unitCellXdistance}
\setlength{\cellYoffset}{#2\unitCellYdistance}
}

\def\labelUnitCell#1#2#3#4#5#6#7#8{
\def\qubitOneLabel{#1}
\def\qubitTwoLabel{#2}
\def\qubitThreeLabel{#3}
\def\qubitFourLabel{#4}
\def\qubitFiveLabel{#5}
\def\qubitSixLabel{#6}
\def\qubitSevenLabel{#7}
\def\qubitEightLabel{#8}
}

\def\drawUnitCell#1#2#3#4#5#6#7#8#9{
\node (#1left1) [qubitGeneric,chain#2] at (\cellXoffset,\cellYoffset-\qubitYdistance*0) {};
\node at (\cellXoffset,\cellYoffset-\qubitYdistance*0) {\qubitOneLabel};
\node (#1left2) [qubitGeneric,chain#3] at (\cellXoffset,\cellYoffset-\qubitYdistance*1) {};
\node at (\cellXoffset,\cellYoffset-\qubitYdistance*1) {\qubitTwoLabel};
\node (#1left3) [qubitGeneric,chain#4] at (\cellXoffset,\cellYoffset-\qubitYdistance*2) {};
\node at (\cellXoffset,\cellYoffset-\qubitYdistance*2) {\qubitThreeLabel};
\node (#1left4) [qubitGeneric,chain#5] at (\cellXoffset,\cellYoffset-\qubitYdistance*3) {};
\node at (\cellXoffset,\cellYoffset-\qubitYdistance*3) {\qubitFourLabel};
\node (#1right1) [qubitGeneric,chain#6] at (\cellXoffset+\qubitXdistance,\cellYoffset-\qubitYdistance*0) {};
\node at (\cellXoffset+\qubitXdistance,\cellYoffset-\qubitYdistance*0) {\qubitFiveLabel};
\node (#1right2) [qubitGeneric,chain#7] at (\cellXoffset+\qubitXdistance,\cellYoffset-\qubitYdistance*1) {};
\node at (\cellXoffset+\qubitXdistance,\cellYoffset-\qubitYdistance*1) {\qubitSixLabel};
\node (#1right3) [qubitGeneric,chain#8] at (\cellXoffset+\qubitXdistance,\cellYoffset-\qubitYdistance*2) {};
\node at (\cellXoffset+\qubitXdistance,\cellYoffset-\qubitYdistance*2) {\qubitSevenLabel};
\node (#1right4) [qubitGeneric,chain#9] at (\cellXoffset+\qubitXdistance,\cellYoffset-\qubitYdistance*3) {};
\node at (\cellXoffset+\qubitXdistance,\cellYoffset-\qubitYdistance*3) {\qubitEightLabel};
\foreach \left in {1,2,3,4} {
	\foreach \right in {1,2,3,4} {
		\draw (#1left\left.east) -- (#1right\right.west);
	}
}
}

\def\drawHorizontalConnections#1#2{
\foreach \index in {1,2,3,4} {
	\draw (#1right\index.north) to[in=170,out=10] (#2right\index.north);
}
}

\def\drawStyledHorizontalConnections#1#2#3{
\foreach \index in {1,2,3,4} {
	\draw[#3] (#1right\index.north) to[in=170,out=10] (#2right\index.north);
}
}

\def\drawVerticalConnections#1#2{
\foreach \index in {1,2,3,4} {
	\draw (#1left\index.west) to[in=100,out=260] (#2left\index.west);
}
}

\def\drawStyledVerticalConnections#1#2#3{
\foreach \index in {1,2,3,4} {
	\draw[#3] (#1left\index.west) to[in=100,out=260] (#2left\index.west);
}
}

\newlength{\horizontalSubfigureSeparator}
\setlength{\horizontalSubfigureSeparator}{0.25cm}

\tikzset{qubitGeneric/.style={minimum size=\minQubitWidth,shape=circle,font=\scriptsize,draw=black},
chain1/.style={fill=blue!10},chain2/.style={fill=blue!25},chain3/.style={fill=blue!50},chain4/.style={fill=blue!65},
chain5/.style={fill=red!10},chain6/.style={fill=red!25},chain7/.style={fill=red!50},chain8/.style={fill=red!75},
chain9/.style={fill=green!10},chain10/.style={fill=green!25},chain11/.style={fill=green!50},chain12/.style={fill=green!75},
chain13/.style={fill=yellow!10},chain14/.style={fill=yellow!25},chain15/.style={fill=yellow!50},chain16/.style={fill=yellow!75},
chain17/.style={fill=brown!10},chain18/.style={fill=brown!25},chain19/.style={fill=brown!50},chain20/.style={fill=brown!75},
chain21/.style={fill=purple!10},chain22/.style={fill=purple!25},
chain24/.style={fill=gray!10},chain25/.style={fill=gray!25},
chain28/.style={fill=orange!10},chain29/.style={fill=orange!25},
chainUnused/.style={fill=white},chainBroken/.style={fill=black}}

\maketitle

\urlstyle{same}


\begin{abstract}
The D-Wave adiabatic quantum annealer solves hard combinatorial optimization problems  leveraging quantum physics. The newest version features over 1000 qubits and was released in August~2015. We were given access to such a machine, currently hosted at NASA Ames Research Center in California, to explore the potential for hard optimization problems that arise in the context of databases.

In this paper, we tackle the problem of multiple query optimization (MQO). We show how an MQO problem instance can be transformed into a mathematical formula that complies with the restrictive input format accepted by the quantum annealer. This formula is translated into weights on and between qubits such that the configuration minimizing the input formula can be found via a process called adiabatic quantum annealing. We analyze the asymptotic growth rate of the number of required qubits in the MQO problem dimensions as the number of qubits is currently the main factor restricting applicability. We experimentally compare the performance of the quantum annealer against other MQO algorithms executed on a traditional computer. While the problem sizes that can be treated are currently limited, we already find a class of problem instances where the quantum annealer is three orders of magnitude faster than other approaches. 
\end{abstract}


\section{Introduction}
\label{introSec}

The database area has motivated a multitude of hard optimization problems that probably cannot be solved in polynomial time. Those optimization problems become harder as data processing systems become more complex. This makes it interesting to explore also unconventional optimization approaches. In this paper, we explore the potential of quantum computing for a classical database-related optimization problem, the problem of multiple query optimization (MQO)~\cite{Sellis1988a}. We were granted a limited amount of computation time on a D-Wave 2X adiabatic quantum annealer, currently hosted at NASA Ames Research Center in California. This device is claimed to exploit the laws of quantum physics~\cite{Boixo2012} in the hope to solve NP-hard optimization problems faster than traditional approaches. The machine supports a very restrictive class of optimization problems while it is for instance not capable of running Shor's algorithm~\cite{Shor1994} for factoring large numbers\footnote{\url{http://www.dwavesys.com/blog/2014/11/response-worlds-first-quantum-computer-buyers-guide}}. We will show how instances of the multiple query optimization problem can be brought into a representation that is suitable as input to the quantum annealer. We also report results of an experimental evaluation that compares the time it takes to solve MQO problems on the quantum annealer to the time taken by algorithms that run on a traditional computer. We believe that this is the first paper featuring an experimental evaluation on a quantum computer ever published in the database community.

The quantum annealer, produced by the Canadian company D-Wave\footnote{\url{http://www.dwavesys.com/}}, uses \textit{qubits} instead of bits. While bits have a deterministic value (either 0 or 1) at each point in time during a computation, a qubit may be put into a superposition of states (0 \textit{and} 1) that would be considered mutually exclusive according to the laws of classical physics. Working with qubits instead of bits could in principle allow faster optimization than on a classical computer~\cite{Aaronson2013}. Thinking of qubit superposition as a specific form of parallelization is certainly simplifying but still gives a first intuition for why this is possible. We provide more explanations on quantum computing and on the quantum annealer in Section~\ref{backgroundSec}.

The quantum annealer that we were experimenting with has a net worth of around 15~million US dollars. This price might make main stream adoption seem illusory in the near-term future. However, the company D-Wave is currently considering flexible provisioning models allowing users to buy computation time instead of the hardware\footnote{\url{http://spectrum.ieee.org/podcast/computing/hardware/dwave-aims-to-bring-quantum-computing-to-the-cloud}}. In this scenario, users would use the machine remotely, in a similar way as we did in our experiments. As near-optimal solutions to hard problems can usually be found within milliseconds (see Section~\ref{experimentalSec}), this provisioning model might allow optimization at an affordable rate per instance. Those are some of the factors that encourage us to explore the potential of quantum computing already at this point in time.

The D-Wave adiabatic quantum annealer has been the subject of controversial discussions in the scientific community. Those discussions have focused on two questions: whether quantum effects play indeed a significant role during the optimization process~\cite{Albash2015, Boixo2012, Boixo2014, Lanting2014, Shin2014a, Smolin2013} and whether the performance is significantly better than the one of classical computers~\cite{Hen2015, King2015, King2015a, Rønnow2014}. Recent publications seem to answer the first question positively~\cite{Albash2015, Boixo2014, Lanting2014}. The answer to the second question depends apparently on the specific class of problems considered, leading for instance to different conclusions for range-limited Ising problems~\cite{King2015a} than for Ising problems without weight limits~\cite{Hen2015}. Solving problem classes that are not natively supported by the quantum annealer requires transformation steps which add a problem class-specific overhead in the problem representation size that might prevent the quantum annealer from solving non-trivial instances due to its limited number of qubits. Our paper adds to the discussion concerning the second question by providing a mapping algorithm and experimental results for a specific database-related optimization problem. 

Prior work on MQO~\cite{Bayir2007, Cosar1993a, Dalvi2003, Dokeroglu, Dokeroglu2015, Dokeroglu2015, Kementsietsidis2008, Le2012, Mistry2000, Roy1999, Sellis1988a, Shim1994} did not consider the potential of quantum computing. Prior publications in the area of quantum computing~\cite{Bian2013, Gaitan2012, Grover1996, Lucas2014, Perdomo-Ortiz2015, Venturelli2015, Williams} did not treat the MQO problem.

\begin{algorithm}[t!]
\renewcommand{\algorithmiccomment}[1]{// #1}
\begin{algorithmic}[1]
\State \Comment{Solves multiple query optimization problem $M$}
\Function{QuantumMQO}{$M$}
\State \Comment{Map MQO problem to logical energy formula}
\State $lef\gets$\Call{LogicalMapping}{$M$}
\State \Comment{Map logical into physical energy formula}
\State $pef\gets$\Call{PhysicalMapping}{$lef$}
\State \Comment{Minimize formula on quantum computer}
\State $b_i\gets$\Call{QuantumAnnealing}{$pef$}
\State \Comment{Transform physical into logical solution}
\State $X_p\gets$\Call{PhysicalMapping$^{-1}$}{$b_i$}
\State \Comment{Transform logical solution to MQO solution}
\State $P_e\gets$\Call{LogicalMapping$^{-1}$}{$X_p$}
\State \Comment{Return best set of query plans to execute}
\State \Return{$P_e$}
\EndFunction
\end{algorithmic}
\caption{How to solve multiple query optimization on an adiabatic quantum annealer.}
\label{mainAlg}
\end{algorithm}

Algorithm~\ref{mainAlg} shows the high-level approach by which we obtain solutions to MQO problem instances from a quantum annealer. The goal of MQO is to select the optimal combination of query plans to execute in order to minimize execution cost for a batch of queries. Given an MQO problem instance $M$, we introduce binary variables $X_p$ for each available query plan $p$ that indicate whether the corresponding plan is executed. We transform the given MQO instance into an energy formula (the term derives from the fact that the quantum annealer translates such formulas into energy levels) on those variables that becomes minimal for a variable assignment representing an optimal solution to the initial MQO problem $M$. We call the variables $X_p$ the logical variables to indicate that they cannot yet be represented by single qubits within the qubit matrix of the quantum annealer.  We call the transformation the \textit{logical mapping} and the resulting formula the \textit{logical energy formula}.

The physical mapping transforms the logical energy formula, defined in the variables $X_p$, into a \textit{physical energy formula} that uses the binary variables $b_i$. Each variable $b_i$ is associated with one specific, physical qubit of the quantum annealer. Finding a value assignment for the variables $b_i$ which minimizes the physical energy formula is an NP-hard problem. We use the quantum annealer to solve it. All other transformations depicted in Algorithm~\ref{mainAlg} have polynomial complexity and are executed on a classical computer. 

Based on the solution returned by the quantum annealer, the value assignment to the variables $b_i$ which minimizes the physical energy formula, we transform the solution to the physical energy formula into a solution to the logical energy formula. Finally, we transform the solution to the logical energy formula into a solution to the original MQO problem which is the optimal set $P_e$ of query plans to execute. 

MQO problems cannot be solved with our approach if the number of qubits required by the physical energy formula exceeds the number of qubits available on the quantum annealer. Albeit doubling the number of qubits compared to the predecessor model, the number of qubits is with slightly over one thousand qubits still very limited on the D-Wave 2X that we experimented with. Correspondingly, the limited number of qubits is in practice the most important factor restricting the size of the problem instances that can be treated with the quantum annealer. For that reason, we analyze the ``complexity'' of our mapping algorithm in terms of the asymptotic growth rate of the number of required qubits as a function of the MQO problem dimensions. This approach is common in the area of quantum annealing~\cite{Lucas2014, Williams}. We find that the number of qubits in the physical energy formula grows quadratically in the number of plans per query and at least linearly in the number of queries.

In our experimental evaluation, we compare our approach based on quantum annealing against classical optimization algorithms executed on traditional computers. We compare against classes of algorithms that have been proposed for MQO in prior publications and include integer linear programming, genetic algorithms, and simple greedy heuristics. While the number of available qubits severely limits the class of non-trivial MQO problems that can be treated efficiently on the quantum annealer, we also find a class of problems where the quantum annealer discovers near-optimal solutions at least 1000 times faster than classical approaches.

In summary, our original scientific contributions in this paper are the following:

\begin{itemize}
\item We map MQO problem instances into a representation that can be solved on a quantum annealer. 
\item We analyze the complexity of our mapping method in terms of the asymptotic number of required qubits as a function of the MQO problem dimensions.
\item We experimentally compare the D-Wave 2X quantum annealer against competing approaches for MQO.
\end{itemize}

The remainder of this paper is organized as follows. In Section~\ref{backgroundSec}, we give a short introduction to quantum computing and the quantum annealer. In Section~\ref{modelSec}, we introduce our formal problem model for MQO. We describe the logical mapping in Section~\ref{mappingSec} and the physical mapping in Section~\ref{embeddingSec}. We formally prove correctness of our mapping and analyze the asymptotic complexity in Section~\ref{analysisSec}. In Section~\ref{experimentalSec}, we evaluate our approach experimentally. We discuss related work in Section~\ref{relatedSec} and conclude in Section~\ref{conclusionSec}.

\section{Quantum Computing}
\label{backgroundSec}

We give a short introduction to quantum computing in general and to the specific realization inside the D-Wave quantum annealer. Our goal is to provide the reader with a rough intuition while we simplify many of the details. A detailed introduction to those complex topics is beyond the scope of this paper and we refer interested readers to specialized publications~\cite{Aaronson2013}.

Quantum mechanics describes physical processes at extremely small scale. The laws of quantum mechanics do not match our intuition since our intuition is formed by the macroscopic world. For instance, extremely small particles may at the same time adopt two states that are mutually exclusive according to our normal intuition. 

Quantum computers~\cite{Aaronson2013} are machines that harness quantum physics to potentially achieve speedups over classical computers. Classical computers use bits that are in either one of two states (1 or 0); quantum computers use qubits that can at the same time be set to 1 \textit{and} to 0, a state that we call \textit{superposition}. This allows quantum computers to explore many alternative computational branches at the same time and there are problems (e.g., prime factorization~\cite{Shor1994}) for which quantum algorithms provide an exponential speedup over the best currently known classical algorithms.

The first commercially available machine claimed to harness quantum effects to speed up optimization is the quantum annealer by D-Wave Systems. In order to use the D-Wave quantum annealer, each optimization problem must be represented as a mathematical function with binary variables. The D-Wave computer aims to find the variable value assignments minimizing the given function. 

More precisely, the D-Wave computer minimizes sums of terms that are either linear or quadratic in the output variables. This problem model corresponds to the quadratic unconstrained binary optimization problem which is NP-hard. The following explanations of the internal workings of the D-Wave machine show that this choice of input format is intrinsically imposed by the D-Wave architecture. 

The D-Wave machine represents binary variables as qubits. Qubits are realized as tiny electric circuits. Those circuits are cooled down to a temperature of 13~millikelvin. Quantum effects appear at this temperature and the current may flow at the same time clockwise and counterclockwise within the circuits, thereby representing qubit superposition. The input function that needs to be minimized is translated into magnetic fields affecting single qubits or qubit pairs. Fields affecting single qubits represent linear terms while fields affecting qubit pairs represent quadratic terms. The strength of those magnetic fields is tuned to be proportional to the weights assigned to the corresponding terms in the input function. Thereby we obtain a physical system that minimizes its total energy for qubit states that represent variable assignments minimizing the input function.

The goal of minimizing the input function is translated into the goal of minimizing the energy level within a physical system in which quantum effects are present. In order to reach the minimal energy level (and thereby solving the input problem), the D-Wave computer executes a process called quantum annealing. 

We introduce quantum annealing informally by contrasting it from the simulated annealing algorithm (SA) which is a classic heuristic optimization algorithm. SA simulates thermal annealing in software while D-Wave performs actual quantum annealing in hardware. Both annealing algorithms process an energy function with the goal to find its global minimum. The SA algorithm performs a set of moves in the search space, using evaluations of the given cost function as guidance in the hope to eventually reach a global minimum. The quantum annealing algorithm starts instead with a simplified cost function whose global minimum can be easily calculated. During optimization, the quantum annealing algorithm does not perform moves in the search space but rather transforms the cost function slowly from the initialization function to the cost function of interest. During that process, the quantum annealer is in a superposition of possible states, unlike its deterministic counterpart. If this transformation is executed slowly enough and without disturbances then the quantum annealer is guaranteed to remain within the global minimum throughout the whole transformation~\cite{Farhi2000} which can be read out after annealing terminates. In practice, annealing runs are often disturbed despite all shielding efforts and a multitude of runs must be executed before finding an optimal solution.

\setlength{\unitCellXdistance}{1.5cm}
\setlength{\unitCellYdistance}{1.5cm}
\setlength{\qubitXdistance}{0.75cm}
\setlength{\qubitYdistance}{0.35cm}
\setlength{\minQubitWidth}{1mm}

\begin{figure}[t!]
\centering
\begin{tikzpicture}
\positionUnitCell{0}{1}
\labelUnitCell{}{}{}{}{}{}{}{}
\drawUnitCell{ul}{1}{1}{1}{1}{1}{1}{1}{1}
\positionUnitCell{0}{0}
\labelUnitCell{}{}{}{}{}{}{}{}
\drawUnitCell{ll}{1}{1}{1}{1}{1}{1}{1}{1}
\positionUnitCell{1}{1}
\labelUnitCell{}{}{}{}{}{}{}{}
\drawUnitCell{ur}{1}{1}{1}{1}{1}{1}{1}{1}
\positionUnitCell{1}{0}
\labelUnitCell{}{}{}{}{}{}{}{}
\drawUnitCell{lr}{1}{1}{1}{1}{1}{1}{1}{1}
\drawHorizontalConnections{ll}{lr}
\drawHorizontalConnections{ul}{ur}
\drawVerticalConnections{ul}{ll}
\drawVerticalConnections{ur}{lr}
\end{tikzpicture}
\caption{Four neighboring unit cells containing eight qubits each, connected in a Chimera structure.\label{matrixFig}}
\end{figure}
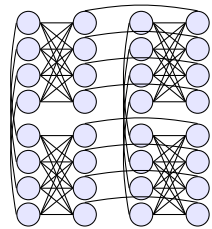

\setlength{\unitCellXdistance}{1.5cm}
\setlength{\unitCellYdistance}{2.25cm}
\setlength{\qubitXdistance}{0.75cm}
\setlength{\qubitYdistance}{0.5cm}
\setlength{\minQubitWidth}{4mm}

We executed our experiments on the D-Wave 2X which was released in August~2015 and is the most recent model in a series of quantum annealers presented by D-Wave with a net price of around 15~million dollars. The number of qubits has been roughly doubling from one model to the next over the past years and the D-Wave 2X features a matrix of 1152 interconnected qubits. The manufacturing process is currently imperfect and only 1097 out of 1152 qubits were fully functional on the machine that we used. Connections between qubits are sparse and form the so called Chimera graph~\cite{Mcgeoch2013a}. Figure~\ref{matrixFig} shows an extract of the Chimera graph structure as it is available on the qubit matrix of the quantum annealer. Qubits are partitioned into so called unit cells. Each unit cell contains eight qubits in two colons and connects each qubit to all four qubits in the opposite colon but not to qubits in the same colon. Qubits in the left colon are connected to their respective counterpart in the qubit cell above and below while qubits in the right colon are connected to their counterparts in the cells to the right and to the left (unless it is the border of the qubit matrix). Each qubit is hence connected to at most six other qubits. The D-Wave 2X uses 144 unit cells.

\section{Formal Model}
\label{modelSec}

We first formalize our problem model of multiple query optimization (MQO). The goal in multiple query optimization is to minimize the joint execution cost for a batch of queries by exploiting possibilities to share computation between different query plans. 

An MQO problem instance is characterized by a set $Q$ of queries. Each query either represents a final result that is requested by the user or an intermediate result that is useful when generating final results. Each query $q\in Q$ is associated with a set of alternative generation plans $P_q$. For a final result, all associated plans must represent methods of generating that result. The generation of intermediate results is optional and the plan set of an intermediate result may contain one plan that represents the possibility of not generating that result.

Set $P=\cup_qP_q$ denotes the set of all considered plans. Each plan $p\in P$ is associated with an execution cost $c_p$. This is the cost of processing the plan without exploiting any previously generated intermediate results. Plans for different queries may however share partial results. It is beneficial to select groups of plans that can share many intermediate results to reduce processing cost. Given two plans $p_1$ and $p_2$ that can share intermediate results, we denote by $s_{p1,p2}>0$ the cost reduction that can be achieved by sharing. Note that our model is not restricted to the case that two plans share an intermediate result. If more than two plans can share an intermediate result, we introduce pair-wise connections between the result and all plans that may share it.

A solution to an MQO problem instance is a subset of plans $P_e\subseteq P$ that are selected for execution. A solution is only valid if exactly one plan is selected for each query and $\forall q\in Q:|P_q\cap P_e|=1$. As discussed before, selecting a plan does not necessarily mean that the corresponding result is generated in case of intermediate results. We denote the accumulated execution cost of a plan set by $C(P_e)=\sum_{p\in P_e}c_p-\sum_{\{p1,p2\}\subseteq P_e}s_{p1,p2}$. A solution is optimal if its execution cost is minimal among all valid solutions. 

We selected an MQO problem model that shortens our following descriptions of the transformation into a quadratic unconstrained binary optimization (QUBO) problem, the formalism that we introduce next. Our MQO model is however equivalent to MQO problem models that were used in prior work\footnote{If each query plan is modeled by a set of tasks~\cite{Sellis1988a} then we make in our model the execution cost of the plan equal to the sum of the execution costs of all tasks and introduce one extra query for each of the tasks with an execution cost equal to the task cost and a cost savings link between task and plan whose value equals the task execution cost again.} and the problem remains NP-hard.


A QUBO problem is defined over a set $\{X_i\}$ of binary variables (with value domain $\{0,1\}$). A solution to a QUBO problem assigns each of the variables to one of the two possible values. The goal is to minimize the following function that depends on the binary variables: $\sum_{i\leq j}w_{ij}X_iX_j$. The weights $w_{ij}$ are problem instance specific. Note that the formula contains linear terms (for $i=j$ since $x_i^2=x_i$ for binary variables) as well as quadratic terms (for $i\neq j$). A solution to a QUBO problem is optimal if it minimizes the function from above among all possible solutions.


\section{Logical Mapping}
\label{mappingSec}

We show now how to transform an MQO problem instance into a QUBO problem instance. This step is required since the quantum annealer can only solve QUBO problems.

As discussed in Section~\ref{modelSec}, an MQO problem is defined by a set $Q$ of queries, a set $P_q$ of plans for each query $q\in Q$ with $P=\cup_q P_q$, execution cost values $c_p$ for each plan $p\in P$, and possible cost savings $s_{p1,p2}$ for each plan pair $p1,p2$. A solution is a subset of plans that are selected for execution such that one plan is selected per query. 

Only binary variables may appear in a QUBO problem. We must therefore represent the solution space of the MQO problem using binary variables. Given a set $P$ of plans, we introduce a binary variable $X_p$ for each plan $p\in P$. If $X_p=1$ then plan $p$ is selected for execution while $p$ is not executed if $X_p=0$.

An MQO solution is only valid if exactly one plan is selected for each query. If our goal was to transform an MQO problem into an integer linear program, we could introduce constraints of the form $\sum_{p\in P_q}X_p=1$ for each $q\in Q$ to guarantee that all returned solutions are valid. Unfortunately, the QUBO formalism does not allow to express constraints directly. As the optimal solution to a QUBO problem minimizes a quadratic formula, we can however add terms to that formula that take high values if constraints are violated. This approach guarantees a valid solution if those terms are scaled up by sufficiently high weights.

We call the quadratic formula defining the QUBO problem short the \textit{energy formula} in the following as it is translated into energy levels by the D-Wave annealer. We decompose the constraint that exactly one plan is selected per query into two parts: we require that at least one plan is selected and that at most one plan is selected. In order to assure that at least one plan is selected for each query, we can simply add the term $E_L=-\sum_{p\in P}X_p$ to the energy formula. As lower values of the energy formula are preferable, this term motivates to set all variables $X_p$ to one. We can express the constraint that at most one plan is selected by adding the term $E_M=\sum_{q\in Q}\sum_{\{p1,p2\}\subseteq P_q}X_{p1}X_{p2}$ to the energy formula. This term takes value zero if at most one plan is selected per query and at least value one otherwise. As we will discuss in the following paragraphs, both terms will have to be scaled by an appropriate factor to make sure that all constraints are respected.

The terms that we have so far inserted into the energy formula make sure that a valid solution is preferable compared to an invalid solution. The goal of the MQO problem is however to minimize execution cost. We must introduce additional energy terms to make a valid solution with lower execution cost preferable over a valid solution with higher execution cost. 

We take into account plan execution cost by introducing the term $E_C=\sum_{p\in P}c_pX_p$ into the energy formula. This means that the execution cost of each selected plan $p$ with $X_p=1$ is added. On the other hand, we must introduce the term $E_S=-\sum_{\{p1,p2\}\subseteq P}s_{p1,p2}X_{p1}X_{p2}$ to represent the possibility of sharing intermediate results between plans. We finally scale up the first two terms that we introduced by a factor whose value we discuss in the following. The resulting energy formula reads\[
w_LE_L+w_ME_M+E_C+E_S.\]

We discuss in the following how to choose the weights $w_L$ and $w_M$. It is crucial to choose the weights as low as possible since having high weights seems to increase the chances of obtaining sub-optimal solutions from the quantum annealer~\cite{King2015a}. We will derive inequalities of the form $w>a$ in the following where $w$ is a weight and $a$ a value that lower-bounds the admissible weights. Having such an inequality, we prefer for the aforementioned reason to choose $w=a+\varepsilon$ in general where $\varepsilon$ is a small value (we typically use $\varepsilon=0.25$ in our implementation).

The energy formula contains two terms that motivate valid solutions ($E_L$ and $E_M$) and two terms that motivate solutions with lower execution cost ($E_C$ and $E_S$). The terms motivating a valid solution should intuitively obtain higher weights than the ones motivating low-cost solutions. If the terms enforcing valid solutions are not associated with sufficiently high weights then the optimal QUBO solution might not select any plans to save execution cost. 

We must make sure that the motivation of selecting at least one plan is always higher than the motivation to save execution cost by not selecting any plan. We accomplish this by requiring $w_L>\max_{p\in P}c_p$. Having scaled up $E_L$ by that factor, the partial energy formula $w_LE_L+E_C+E_S$ would be minimized by executing each plan for each query. This does clearly not reflect the original MQO problem and we must add $w_ME_M$ to restrict the number of plan selections per query to one.

Clearly we must choose $w_M>w_L$ to accomplish the aforementioned goal. This is however insufficient. The generation cost of a query can be lower than the cost reduction achievable by sharing it among other plans. Hence, even if we have $w_M>w_L$, the energy formula might still be minimized by executing multiple plans for the same query. This is due to a shortcoming of the QUBO representation: the QUBO representation leads to believe that it is possible to accumulate cost savings by generating the same result according to multiple plans. In reality, this is of course not the case. We circumvent that problem by explicitly enforcing that at most one plan is selected per query. This is guaranteed if $c_M>c_L+\max_{p1\in P}\sum_{p2\in P}s_{p1,p2}$.

\begin{example}
We show how to transform a simple MQO problem into the QUBO representation. Assume that four plans $p_1$, $p_2$, $p_3$, and $p_4$ are considered with execution cost 2, 4, 3, 1 respectively. The first two plans generate query $q_1$ and the next two plans generate query $q_2$. Assume further that $p_2$ and $p_3$ can share an intermediate result allowing cost savings of 5 cost units. The QUBO representation uses the binary variables $X_1$, $X_2$, $X_3$, and $X_4$ that are associated with plans $p_1$ to $p_4$ and are set to one if the corresponding plan is executed. Then execution cost is represented by the term $E_C=2X_1+4X_2+3X_3+1X_4$. Potential cost savings are represented by the term $E_S=-5X_2X_3$. The term $E_L=-\sum_{i=1}^4X_i$ enforces at least one plan selection for each of the two queries and is weighted by factor $w_L=4+\varepsilon$. Term $E_M=X_1X_2+X_3X_4$ enforces at most one plan selection if weighted by factor $w_M=w_L+5$. The variable assignment $X_1=0$, $X_2=1$, $X_3=1$, $X_4=0$ minimizes the energy formula and represents the optimal solution to the MQO problem at the same time.
\end{example}

We prove formally in Section~\ref{analysisSec} that the mapping method presented in this section is correct.

\begin{figure*}[t!]
\centering
\subfigure[TRIAD with 5 chains\label{5varTRIADSub}]{
\hspace{0.5cm}
\begin{tikzpicture}
\positionUnitCell{0}{0}
\labelUnitCell{1}{3}{4}{5}{2}{3}{4}{5}
\drawUnitCell{1}{1}{3}{4}{5}{2}{3}{4}{5}
\end{tikzpicture}
\hspace{0.5cm}
}
\hspace{\horizontalSubfigureSeparator}
\subfigure[TRIAD with 8 chains\label{8varTRIADSub}]{
\begin{tikzpicture}
\positionUnitCell{0}{1}
\labelUnitCell{1}{2}{3}{4}{1}{2}{3}{4}
\drawUnitCell{ul}{1}{2}{3}{4}{1}{2}{3}{4}
\positionUnitCell{0}{0}
\labelUnitCell{1}{2}{3}{4}{5}{6}{7}{8}
\drawUnitCell{ll}{1}{2}{3}{4}{5}{6}{7}{8}
\positionUnitCell{1}{0}
\labelUnitCell{5}{6}{7}{8}{5}{6}{7}{8}
\drawUnitCell{r}{5}{6}{7}{8}{5}{6}{7}{8}
\drawHorizontalConnections{ll}{r}
\drawVerticalConnections{ul}{ll}
\end{tikzpicture}
}
\hspace{\horizontalSubfigureSeparator}
\subfigure[TRIAD with 12 chains\label{12varTRIADSub}]{
\begin{tikzpicture}
\positionUnitCell{0}{2}
\labelUnitCell{1}{2}{3}{4}{1}{2}{3}{4}
\drawUnitCell{l2}{1}{2}{3}{4}{1}{2}{3}{4}
\positionUnitCell{0}{1}
\labelUnitCell{1}{2}{3}{4}{5}{6}{7}{8}
\drawUnitCell{l1}{1}{2}{3}{4}{5}{6}{7}{8}
\positionUnitCell{0}{0}
\labelUnitCell{1}{2}{3}{4}{9}{10}{11}{12}
\drawUnitCell{l0}{1}{2}{3}{4}{9}{10}{11}{12}
\positionUnitCell{1}{1}
\labelUnitCell{5}{6}{7}{8}{5}{6}{7}{8}
\drawUnitCell{m1}{5}{6}{7}{8}{5}{6}{7}{8}
\positionUnitCell{1}{0}
\labelUnitCell{5}{6}{7}{8}{9}{10}{11}{12}
\drawUnitCell{m0}{5}{6}{7}{8}{9}{10}{11}{12}
\positionUnitCell{2}{0}
\labelUnitCell{9}{10}{11}{12}{9}{10}{11}{12}
\drawUnitCell{r}{9}{10}{11}{12}{9}{10}{11}{12}
\drawHorizontalConnections{l0}{m0}
\drawHorizontalConnections{m0}{r}
\drawHorizontalConnections{l1}{m1}
\drawVerticalConnections{l2}{l1}
\drawVerticalConnections{l1}{l0}
\drawVerticalConnections{m1}{m0}
\end{tikzpicture}
}
\hspace{\horizontalSubfigureSeparator}
\subfigure[TRIAD with 12 chains and two broken qubits\label{brokenQubitsSub}]{
\begin{tikzpicture}
\positionUnitCell{0}{2}
\labelUnitCell{1}{2}{}{4}{1}{2}{}{4}
\drawUnitCell{l2}{1}{2}{Unused}{4}{1}{2}{Unused}{4}
\positionUnitCell{0}{1}
\labelUnitCell{1}{2}{}{4}{5}{6}{7}{8}
\drawUnitCell{l1}{1}{2}{Broken}{4}{5}{6}{7}{8}
\positionUnitCell{0}{0}
\labelUnitCell{1}{2}{}{4}{}{10}{11}{12}
\drawUnitCell{l0}{1}{2}{Unused}{4}{Broken}{10}{11}{12}
\positionUnitCell{1}{1}
\labelUnitCell{5}{6}{7}{8}{5}{6}{7}{8}
\drawUnitCell{m1}{5}{6}{7}{8}{5}{6}{7}{8}
\positionUnitCell{1}{0}
\labelUnitCell{5}{6}{7}{8}{}{10}{11}{12}
\drawUnitCell{m0}{5}{6}{7}{8}{Unused}{10}{11}{12}
\positionUnitCell{2}{0}
\labelUnitCell{}{10}{11}{12}{}{10}{11}{12}
\drawUnitCell{r}{Unused}{10}{11}{12}{Unused}{10}{11}{12}
\drawHorizontalConnections{l0}{m0}
\drawHorizontalConnections{m0}{r}
\drawHorizontalConnections{l1}{m1}
\drawVerticalConnections{l2}{l1}
\drawVerticalConnections{l1}{l0}
\drawVerticalConnections{m1}{m0}
\end{tikzpicture}
}
\caption{TRIAD pattern in different sizes: we show qubits as circles, annotated by the ID of the logical variable that they represent. The mapping from variables to qubits assures that each variable shares at least one connection (in black) with each of the other variables.\label{TRIADFig}}
\end{figure*}
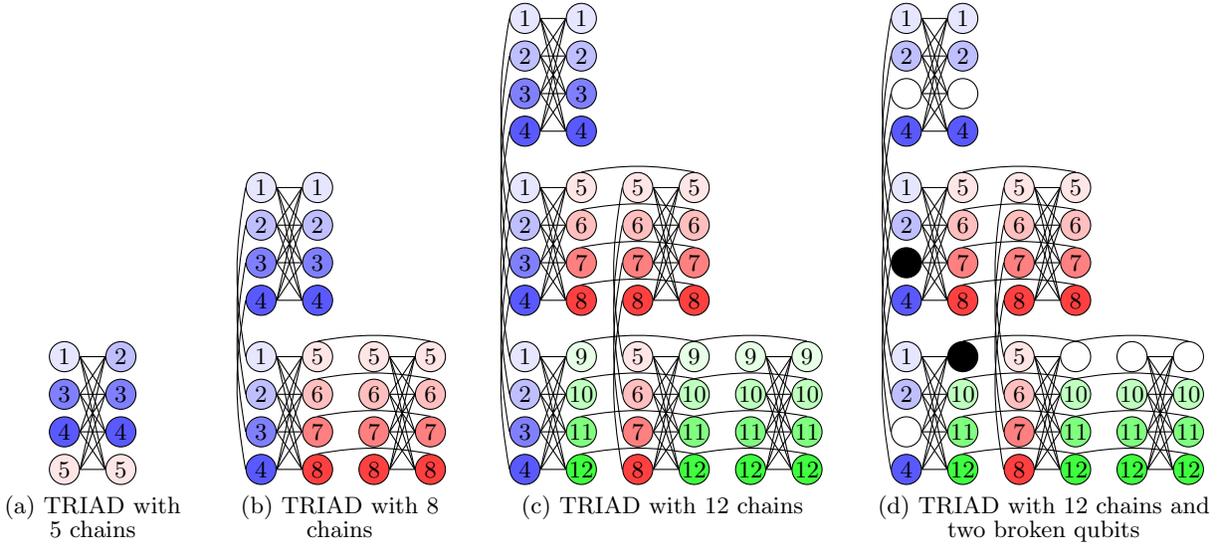

\section{Physical Mapping}
\label{embeddingSec}

We have seen in the previous section how to transform an MQO problem into an energy formula defined on the variables $X_p$. We require one more transformation until we can apply the quantum annealer: we must choose for each logical variable $X_p$ a group of physical qubits to represent it. Then we must set the weights on single qubits and the strengths of the couplings between qubits in order to translate the \textit{logical energy formula} of the form $\sum_{\{p1,p2\}\subseteq P}w_{p1,p2}X_{p1}X_{p2}$ into a \textit{physical energy formula} of the form $\sum_{i\leq j}\widetilde{w}_{ij}b_{i}b_{j}$ where $b_i$ represents the state of the $i$-th qubit of the quantum annealer. We call this transformation the \textit{physical mapping} or \textit{embedding}.

This second transformation is required since it is in general insufficient to represent one QUBO variable by one qubit. This is due to the sparse connection structure between qubits (see Section~\ref{backgroundSec} for a detailed description of the connection structure). Each qubit is connected to at most six other qubits. In the physical energy formula, only the weights between connected qubits can be different from zero. Hence for a fixed $i$ there are at most six values for $j$ such that $\widetilde{w}_{ij}\neq0$. If a QUBO variable interacts with more than six other QUBO variables (meaning that for a fixed plan $p1$ there are more than six plans $p2$ such that $w_{p1,p2}\neq0$ in the logical energy formula) then we must represent that variable by multiple qubits. 

The physical mapping consists of three steps. First, for each QUBO variable we must select a group of physical qubits to represent it. Second, if the energy formula contains a term of the form $w_iX_i$ where $w_i$ is a weight and $X_i$ a QUBO variable then we must distribute that weight over all qubits representing $X_i$: if $B$ denotes the set of qubits representing $X_i$ then the weight $w_i/|B|$ is added on each qubit in $B$. If a term $w_{ij}X_iX_j$ appears in the energy formula then we select one qubit $b_1$ among the qubits representing $X_i$ and another qubit $b_2$ among the qubits representing $X_j$ such that $b_1$ and $b_2$ are connected by a coupling in the qubit matrix and we increase the strength of that coupling by $w_{ij}$. In a third step, we must make sure that all qubits representing the same variable ``behave as one bit'' and are assigned consistently to the same value after an annealing run. We accomplish this by adding additional weights on the qubits and on the couplings between qubits representing the same variable such that the minimum energy is reached for a consistent assignment. This requires that all qubits representing the same variable form a chain of connected qubits. 

We provide further details on step one and step three in the following, starting with step one.

Mapping variables to qubits is a highly non-trivial problem as the mapping must satisfy various constraints. First, we must represent variables by groups of qubits that are connected in a chain. Second, if two logical variables appear together in a quadratic term in the energy formula then the two groups of qubits representing those variables must be connected, i.e.\ at least one qubit from the first group is connected to at least one qubit of the second group. Third, we must take into account that some of the qubits and inter-qubit connections on the D-Wave annealer are broken and cannot be used (see Figure~\ref{matrixFig}). 

Finding for a given QUBO problem the embedding that satisfies all of the aforementioned constraints while consuming the minimal number of qubits is an NP-hard problem~\cite{Klymko2014}. We cannot solve it optimally without the risk that the time for finding an optimal embedding dominates the time of finding the optimal solution to the resulting QUBO problem. For that reason, we currently use simple embedding schemes that can be generated with negligible time overhead and are presented in the following. 

Figure~\ref{TRIADFig} shows the TRIAD pattern proposed by Choi~\cite{Choi2011} in the graphical representation introduced by Venturelli et al.~\cite{Venturelli2014}. This pattern allows to embed arbitrary QUBO problems. Figures~\ref{5varTRIADSub} to \ref{12varTRIADSub} show the pattern in different sizes, supporting 5, 8, and 12 logical variables. When representing each logical variable by a chain of qubits in this pattern (qubits in the same chain are labeled by the same number in the figure) then arbitrary energy formulas can be modeled since the pattern connects each pair of variables. 

The method currently used for manufacturing the qubit matrix is imperfect and results in a certain percentage of broken qubits. If a qubit chain contains broken qubits then the entire chain becomes unusable since it cannot be guaranteed anymore that all qubits in the chain are assigned to the same value. Figure~\ref{brokenQubitsSub} illustrates the problem, visualizing broken qubits in black and intact qubits in unusable chains in white. 

All chains in the TRIAD pattern are connected by at least one coupling. The downside of enabling so many connections is that the number of qubits consumed by the TRIAD pattern grows quadratically in the number of chains and qubits must be considered a scarce resource on current quantum annealers. Analyzing the energy formula from the last section, we find that we require connections between logical variables representing different plans for the same query (due do the quadratic sub-terms contained in $E_M$) and connections between variables representing plans for different queries with work overlap (due to the terms in $E_S$). 

Existing approaches for MQO cluster queries based on structural properties in a preprocessing step~\cite{Le2012} such that queries in different clusters are less likely to share intermediate results. We can exploit such a clustering in certain cases as illustrated in Figure~\ref{alternativePatternFig}: instead of a single TRIAD pattern, we use multiple TRIAD patterns where each TRIAD represents all variables associated with the plans for the query in one single cluster. As different plans for the same query are integrated into the same TRIAD structure, we are sure to realize all connections required by term $E_M$. As plans for different queries in the same cluster are integrated into the same TRIAD as well, all connections required by term $E_S$ can be realized, too. The connections between qubits representing plans in different clusters are sparse but so are the opportunities of work sharing between them and connections between plans in different clusters can only represent work sharing opportunities. The advantage of using the clustered pattern over the single TRIAD pattern is that the number of required qubits grows more slowly in the number of queries and plans as we analyze in more detail in the following section.

\setlength{\unitCellXdistance}{1.375cm}

\begin{figure}[t!]
\begin{tikzpicture}
\positionUnitCell{0}{1}
\labelUnitCell{1}{2}{3}{4}{1}{2}{3}{4}
\drawUnitCell{c0u}{2}{2}{2}{2}{2}{2}{2}{2}
\positionUnitCell{0}{0}
\labelUnitCell{1}{2}{3}{4}{5}{6}{7}{8}
\drawUnitCell{c0l}{2}{2}{2}{2}{2}{2}{2}{2}
\positionUnitCell{1}{0}
\labelUnitCell{1}{2}{3}{4}{5}{6}{7}{8}
\drawUnitCell{c1l}{2}{2}{2}{2}{2}{2}{2}{2}
\positionUnitCell{1}{1}
\labelUnitCell{1}{2}{3}{4}{1}{2}{3}{4}
\drawUnitCell{c1u}{6}{6}{6}{6}{6}{6}{6}{6}
\positionUnitCell{2}{1}
\labelUnitCell{5}{6}{7}{8}{1}{2}{3}{4}
\drawUnitCell{c2u}{6}{6}{6}{6}{6}{6}{6}{6}
\positionUnitCell{2}{0}
\labelUnitCell{5}{6}{7}{8}{5}{6}{7}{8}
\drawUnitCell{c2l}{6}{6}{6}{6}{6}{6}{6}{6}
\positionUnitCell{3}{1}
\labelUnitCell{1}{2}{3}{4}{1}{2}{3}{4}
\drawUnitCell{c3u}{10}{10}{10}{10}{10}{10}{10}{10}
\positionUnitCell{3}{0}
\labelUnitCell{1}{2}{3}{4}{5}{6}{7}{8}
\drawUnitCell{c3l}{10}{10}{10}{10}{10}{10}{10}{10}
\positionUnitCell{4}{0}
\labelUnitCell{1}{2}{3}{4}{5}{6}{7}{8}
\drawUnitCell{c4l}{10}{10}{10}{10}{10}{10}{10}{10}
\positionUnitCell{4}{1}
\labelUnitCell{1}{2}{3}{4}{1}{2}{3}{4}
\drawUnitCell{c4u}{14}{14}{14}{14}{14}{14}{14}{14}
\positionUnitCell{5}{1}
\labelUnitCell{5}{6}{7}{8}{1}{2}{3}{4}
\drawUnitCell{c5u}{14}{14}{14}{14}{14}{14}{14}{14}
\positionUnitCell{5}{0}
\labelUnitCell{5}{6}{7}{8}{5}{6}{7}{8}
\drawUnitCell{c5l}{14}{14}{14}{14}{14}{14}{14}{14}
\drawHorizontalConnections{c0u}{c1u}
\drawHorizontalConnections{c1l}{c2l}
\drawHorizontalConnections{c2u}{c3u}
\drawHorizontalConnections{c2l}{c3l}
\drawHorizontalConnections{c3u}{c4u}
\drawHorizontalConnections{c4l}{c5l}
\drawVerticalConnections{c1u}{c1l}
\drawVerticalConnections{c4u}{c4l}
\drawHorizontalConnections{c0l}{c1l}
\drawHorizontalConnections{c1u}{c2u}
\drawHorizontalConnections{c3l}{c4l}
\drawHorizontalConnections{c4u}{c5u}
\drawVerticalConnections{c0u}{c0l}
\drawVerticalConnections{c2u}{c2l}
\drawVerticalConnections{c3u}{c3l}
\drawVerticalConnections{c5u}{c5l}
\end{tikzpicture}
\caption{Clustered embedding pattern: qubits representing plans in different clusters are distinguished by their color (four colors hence four clusters), the qubit label is the plan identifier (numbers one to eight represent eight alternative plans per cluster). \label{alternativePatternFig}}
\end{figure}
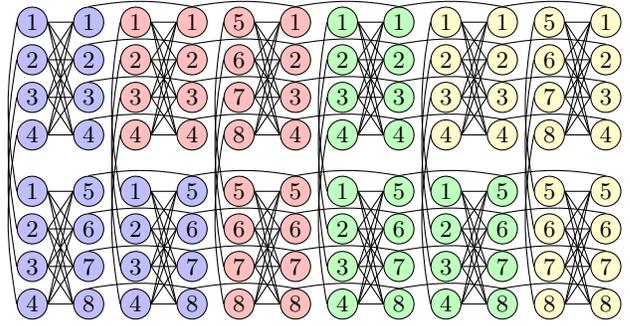




The annealing process that is executed by the quantum annealer takes only into account the physical energy formula. We cannot directly integrate the information that multiple qubits represent the same logical variable and should be assigned to the same value. Instead, we must add more terms to the physical energy formula that make groups of qubits ``behave as one bit''. More precisely, we add, for each group of qubits, terms to the energy formula that take high values if the qubits are assigned to different values. As the goal is to minimize the energy formula, such terms will drive the annealing process towards solutions that assign groups of qubits representing the same variable to the same value. Then we can read out one single value for the entire qubit group and associate it with the represented variable.

Assume that two connected qubits $b_1$ and $b_2$ represent the same variable. We motivate assigning the same value to both of them by adding the term $b_1+b_2-2b_1b_2$ to the energy formula. This term takes value one if the qubits are assigned to different values and takes value zero if both qubits are assigned to the same value. 

Assume now that we have a group of qubits with more than two elements that need to be assigned to the same value. We generally require qubit groups representing the same variable to form a chain. This means that we can order the qubits into a sequence $\langle b_1,b_2,\ldots,b_m\rangle$ such that each qubit $b_i$ is connected to its successor $b_{i+1}$ in the qubit matrix. Under this assumption, we can add energy terms of the form $E_B(i)=b_i+b_{i+1}-2b_ib_{i+1}$ that motivate assigning the same value to two consecutive qubits. Adding the terms $E_B=\sum_{i=1}^{m-1}E_B(i)$ to the energy formula motivates assigning all qubits to the same value.

We assume in the following that the terms from the logical energy formula have already been integrated into the physical energy formula as described under step two at the beginning of this section. Hence the physical energy formula contains terms in addition to the terms $E_B$. This means that we have to scale up the terms $E_B$ by a factor that is sufficiently high to assure that the energy formula becomes minimal for a value assignment where all equality constraints, represented by $E_B$, are satisfied. As discussed in Section~\ref{mappingSec}, we choose the scaling factors as low as possible to avoid a large range of possible energy values.

The following scaling method is based on ideas by Choi~\cite{Choi2008}. We treat each group $B$ of qubits representing the same variable separately and calculate a specific scaling factor for $E_B$. This scaling factor must make sure that a solution with inconsistent assignments for a qubit group improves (i.e., the value of the energy formula decreases) once replacing inconsistent assignments by a consistent one (either all qubits in the group are set to one or all are set to zero). Consider a group $B$ of qubits representing the same variable that are assigned to inconsistent values. The term $w_BE_B$ adds at least $w_B$ to the energy formula as the chain must be broken at least at one position. Replacing the inconsistent assignment by a consistent assignment lets $w_BE_B$ take the value zero and reduces the total energy level by $w_B$.

Making the assignment for $B$ consistent must reduce the energy value of $w_BE_B$ but it might increase the value of other terms in the energy formula. We calculate in the following the upper bound $U$ on the increase in the other energy terms. We denote by $U_{0\rightarrow1}(b)$ the maximal increase in energy caused by changing the value of $b$ from zero to one. Denote by $v$ the weight on $b$ and by $v_i$ all weights on couplings that connect $b$ to qubits outside of $B$. Then we have $U_{0\rightarrow 1}(b)=v+\sum_i\max(v_i,0)$. We pessimistically assume here that each qubit connected to $b$ via a positive weight is set to one while qubits connected via a negative weight are set to zero. This yields an upper bound on the increase in energy. We can calculate an upper bound for the energy increase when setting the value of $b$ from one to zero in the analogue fashion and denote the result by $U_{1\rightarrow0}(b)$.

We have the choice between setting all qubits in $B$ to one or setting all of them to zero in order to make the assignment for $B$ consistent. We can select the option that leads to a lower increase in energy and therefore obtain $U=\min(\sum_{b\in B}U_{1\rightarrow0}(b),\sum_{b\in B}U_{0\rightarrow1}(b))$ as upper bound for the increase in all energy terms except for $E_B$ when making an inconsistent assignment consistent. This means that the energy formula must become minimal for a consistent assignment if we set $w_B=U+\varepsilon$.

\section{Formal Analysis}
\label{analysisSec}

In this section, we prove that the transformation from MQO to QUBO problems that we introduced in Section~\ref{mappingSec} is correct, meaning that the optimal solution to the QUBO problem represents indeed the optimal solution to the MQO problem. Later, we will analyze how the number of qubits that our mapping requires evolves as a function of the MQO problem dimensions. 

We prove that the energy formula $w_ME_M+w_LE_L+E_C+E_S$ becomes minimal for an assignment of variables to values representing an optimal solution to the MQO problem from which the energy formula was derived.

\begin{lemma}
The energy formula is minimized by selecting at most one plan per query.\label{atMostOneLemma}
\end{lemma}
\begin{proof}
Assume that the energy formula was minimized by setting $X_{p1}=X_{p2}=1$ where $p1$ and $p2$ are alternative plans for the same query. If we set $X_{p1}=0$ (or $X_{p2}$) then the value of term $E_C$ decreases by the execution cost of $p1$ while $E_S$ might increase as cost savings enabled by executing $p1$ cannot be realized. Term $E_S$ increases at most by the accumulated cost savings enabled by $p1$ which is $\sum_{p\in P}s_{p1,p}$. The value of term $w_LE_L$ increases by $w_L$. Term $E_M$ contains the sub-term $w_MX_{p1}X_{p2}$ so the value of $E_M$ decreases by $w_M$. In summary, the energy increases at most by $w_L+\sum_{p\in P}s_{p1,p}$ while it decreases by $w_M$ and $w_M>w_L+\sum_{p\in P}s_{p1,p}$. The energy decreases by setting $X_{p1}=0$ which contradicts our initial assumption.
\end{proof}

\begin{lemma}
The energy formula is minimized by selecting at least one plan per query.\label{atLeastOneLemma}
\end{lemma}
\begin{proof}
Assume that the energy formula was minimized by setting $X_p=0$ for all $p\in P_q$ for a query $q\in Q$. Pick one arbitrary plan $p\in P_q$ and set $X_p=1$ instead. Then the value of term $E_C$ increases by the execution cost $c_p$ of that plan. The value of term $E_S$ can only decrease since executing $p$ might enable possibilities to share work and reduce execution cost. The value of $E_M$ remains constant while the value of $w_LE_L$ decreases by $w_L$. In summary, the energy increases at most by $c_p$ and decreases by $w_L$ and $w_L>c_p$. The energy decreases by setting $X_{p}=1$ which contradicts our initial assumption.
\end{proof}

\begin{theorem}
The energy formula is minimized for a valid solution with minimal execution cost.
\end{theorem}
\begin{proof}
The energy formula becomes minimal for a valid solution (meaning that one plan is selected per query) according to Lemmata~\ref{atMostOneLemma} and \ref{atLeastOneLemma}. Furthermore, terms $E_L$ and $E_M$ have the same value for each valid solution (value $-|Q|$ for $E_L$ and value $0$ for $E_M$). This means that those two terms do not influence the choice between valid solutions. The selection of an optimal solution is entirely governed by the combined term $E_C+E_S$. The theorem follows since that term represents exactly the execution cost, taking into account cost reductions by shared work. 
\end{proof}

It is common to analyze the time complexity of optimization algorithms on traditional (non-quantum) computers. It would be interesting to analyze the asymptotic run time until the quantum annealer finds optimal or near-optimal solutions as a function of the MQO problem dimensions. A theoretical framework for analyzing worst case time complexity on adiabatic quantum computers is however currently not available~\cite{Venturelli2015}. This is why the analysis of adiabatic quantum approaches usually focuses on the number of required qubits. This metric is relevant since the number of available qubits imposes most restrictions in practice. 

We analyze the required number of qubits as a function of the following variables. We denote by $n$ the number of query clusters, by $m$ the number of queries per cluster, and by $l$ the number of alternative plans per query. We generally assume that connections between plans in the same cluster are relatively dense while connections between different clusters are relatively sparse. In order to simplify the following analysis, we assume now the extreme case that all plans in each cluster are connected while no connections exist between different clusters. This would allow to decompose the problem and treat different clusters separately but the following results still apply to sparsely connected clusters in which decomposition is not possible. We first analyze the QUBO representation from Section~\ref{mappingSec} in terms of how many qubits it \textit{minimally} requires.

\begin{theorem}
The QUBO representation introduced in Section~\ref{mappingSec} requires $\Omega(n\cdot(m\cdot l)^2)$ qubits.
\end{theorem}
\begin{proof}
The total number of considered plans is $n\cdot m\cdot l$. This is at the same time the number of logical variables and hence a lower bound on the number of qubits (as each variable must be represented at least by one qubit). We must however take into account that each qubit is connected to at most six other qubits. The number of required connections between logical variables leads therefore to another lower bound on the required number of qubits. 

Only quadratic terms in the energy formula require connected qubits. Terms $E_L$ and $E_C$ contain no quadratic sub-terms while $E_M$ connects all plans for the same query and term $E_S$ connects plans with work overlap. As mentioned before, we simplify by assuming that $E_S$ connects all plans in the same cluster but no plans in different clusters. Hence plans are connected to all plans in the same cluster. The number of plans per cluster is $m\cdot l$ so each plan is connected to $\Omega(m\cdot l)$ other plans. Due to the constant number of connections per qubit, this means that each plan must be represented by $\Omega(m\cdot l)$ qubits. Multiplying by the total number of plans, $n\cdot m\cdot l$, yields the postulated result.
\end{proof}

The minimal number of qubits is a property of the logical mapping presented in Section~\ref{mappingSec}. Now we analyze the actual (asymptotic) number of qubits required by the clustered mapping pattern presented in Section~\ref{embeddingSec}. We assume that all qubits used by the pattern are intact.

\begin{theorem}
The physical mapping pattern introduced in Section~\ref{embeddingSec} requires $\Theta(n\cdot (m\cdot l)^2)$ qubits.
\end{theorem}
\begin{proof}[Sketch]
The plans in each cluster are mapped to a TRIAD pattern. We can prove by induction that the number of qubits required by a TRIAD grows quadratically in the number of chains. The number of plans per cluster is $m\cdot l$ so the number of qubits per cluster is in $\Theta((m\cdot l)^2)$. Multiplying by the number of clusters yields the result.
\end{proof}

We see that the asymptotic number of qubits required by our physical mapping matches the lower bound. We finally analyze the time complexity of the preprocessing phase that is executed on a classical computer.

\begin{theorem}
Calculating the physical energy formula is in $O(n\cdot(m\cdot l)^2)$ time.
\end{theorem}
\begin{proof}
We first analyze the time complexity of the logical mapping. The energy formula consists of the terms $E_L$, $E_M$, $E_C$, and $E_S$. The time complexity for calculating the weights for those terms is proportional to the number of linear and quadratic sub-terms plus the complexity of calculating scaling factors. The terms $E_L$ and $E_C$ contain $n\cdot m\cdot l$ sub-terms respectively, term $E_M$ contains $O(n\cdot m\cdot l^2)$ sub-terms, and term $E_S$ contains $O(n\cdot (m\cdot l)^2)$ sub-terms. Calculating $w_L$ requires to determine the maximum out of $n\cdot m\cdot l$ cost values, calculating $w_M$ based on $w_L$ requires to determine the maximum out of $n\cdot m\cdot l$ sums over $O(m\cdot l)$ cost saving values with complexity $O(n\cdot (m\cdot l)^2)$. The total time complexity of the logical mapping phase is $O(n\cdot (m\cdot l)^2)$.

Now we analyze the complexity of the physical mapping. Due to the regularity of the employed patterns, identifying the qubits associated with a logical variable takes linear time in the number of qubits. Weights on and between qubits can be added in constant time per weight. Calculating the scaling factor $w_B$ for a group $B$ of qubits requires to examine the connections of each qubit in $B$. As each qubit is connected to a constant number of other qubits, the time for calculating $w_B$ is linear in $|B|$. In summary, we must calculate scaling factors for $O(n\cdot(m\cdot l)^2)$ qubits and assign $O(n\cdot (m\cdot l)^2)$ weights. The combined complexity of logical and physical mapping is $O(n\cdot(m\cdot l)^2)$.
\end{proof}

We find that the time complexity of the transformation from MQO problems into qubit weight assignments is a low-order polynomial in the MQO problem dimensions. We have not taken into account the complexity of clustering queries, generating alternative plans for each query, and identifying work overlap. This pre-processing step is however required by other MQO optimization methods as well~\cite{Le2012} and its implementation is orthogonal to the selection of optimal plan combinations. If it is initially unclear how many clusters are required then the mapping algorithm can be invoked iteratively for a decreasing number of clusters until the mapping is successful (which is assured for one cluster). Then the time complexity is multiplied by the number of iterations. 



\section{Experimental Evaluation}
\label{experimentalSec}

We were granted a limited amount of computation time on a D-Wave 2X adiabatic quantum annealer with over 1000 qubits that is currently located at NASA Ames Research Center in California. We evaluated its performance on MQO problem instances that have been transformed into mathematical formulas as described before.

Our current approach transforms one MQO problem instance into one QUBO problem instance while we might consider approaches mapping one MQO problem into series of QUBO problems in future work. The size of the problems that can be treated by our current approach is inherently limited by the number of available qubits. The formulas established in the last section can be used to calculate the limits on the MQO problem dimensions until which our approach is applicable. It is clear, without performing any experiments, that there are classes of MQO problems that can be treated by existing MQO algorithms (e.g., 500 queries with three plans or more per query~\cite{Bayir2007}) but not on a quantum annealer with 1097 qubits. This is why we focus our experiments on the opposite question: are there also classes of MQO problems where finding the optimal solution requires non-negligible optimization time on commodity computers and where the quantum annealer outperforms existing approaches? 

This question is interesting since a positive answer would constitute evidence that future models of the quantum annealer with more qubits can become an interesting alternative to classical MQO optimizers and the number of qubits has so far been steadily doubling from one model to the next. The question is also non-trivial and experiments are required to answer it: while absolute optimization times are expected to be lower for the quantum annealer than for commodity computers when optimizing the problem class that is natively supported by the quantum annealer~\cite{Rønnow2014}, the blowup in problem representation size during logical and physical mapping might in principle offset that advantage.

We answer the aforementioned question in the following. Section~\ref{experimentalSetupSub} describes our experimental setup while Section~\ref{experimentalResultsSub} describes and discusses our experimental results.

\subsection{Experimental Setup}
\label{experimentalSetupSub}

We use a D-Wave 2X quantum annealer as described in Section~\ref{backgroundSec}. We use the default time of 129 microseconds per annealing run and 247 microseconds per read-out such that an annealing run with following readout takes 376 microseconds. For each test case, we perform 1000 annealing runs that are partitioned into 10 batches of 100 annealing runs per gauge transformation. A gauge transformation~\cite{Boixo2012} selects for each qubit the physical state representing a one randomly between the two available states. Using multiple gauge transformations reduces the effect of small biases favoring one qubit state over another. 


We compare our approach based on quantum annealing against other optimization algorithms that have been recently proposed for MQO: integer linear programming~\cite{Dokeroglu}, genetic algorithms~\cite{Bayir2007}, and iterated hill climbing~\cite{Dokeroglu2015}. We compare against a commercial integer linear programming solver that we use in two ways: we use it to solve MQO problems directly and we use it to minimize the energy formula that the quantum annealer minimizes, too. We use a linear reformulation of the quadratic energy formula that is more suitable for integer programming solvers~\cite{Dash2013}.

Our heuristic algorithms are implemented in Java (while the integer linear programming solver is implemented in C). We use the Java Genetic Algorithms Package\footnote{\url{http://jgap.sourceforge.net/}} in version 3.6.3 with the default configuration which is a genetic algorithm with single point crossover and a top-n selection strategy. The crossover rate is 0.35 and the mutation rate 1/12. We try different population sizes in our experiments. Our hill climbing algorithm iteratively generates plan selections randomly and improves them via hill climbing until a local optimum is reached. We follow good practices for benchmarking Java  programs\footnote{\url{http://www.ibm.com/developerworks/library/j-benchmark1/}} and execute a code warmup of at least 10 seconds for each algorithm before starting the actual benchmark. All Java-based algorithms were implemented in Java~1.7 and executed using the Java HotSpot(TM) 64-Bit Server Virtual Machine version on an iMac with  i5-3470S 2.90GHz CPU and 16~GB of DDR3 RAM. 

We focus on the core optimization problem and neither perform common sub-expression identification nor query clustering. We consider test cases that map well to the quantum annealer for the reasons outlined before. We vary the number of alternative plans per query and generate test cases for the maximal number of queries that can be represented with the available number of qubits. Each query forms one cluster. The weights between qubits representing different plans for the same query are determined by our mapping scheme. The weights between qubits representing plans of different queries represent cost savings and are chosen randomly. 


\subsection{Experimental Results}
\label{experimentalResultsSub}

\def\timeCostPlot#1#2#3{
\nextgroupplot[legend to name=timeCostLegend, legend entries={LIN-MQO, LIN-QUB, QA, CLIMB, GA(50), GA(200)}, legend columns=3]
\addplot[mark=triangle, mark size=2, draw=blue] table[header=true, col sep=comma] {plotsdata/linear/Q#1P#2IfalseT#3_scaledCostCurve};
\addplot[mark=o, mark size=1, draw=blue] table[header=true, col sep=comma] {plotsdata/qubolinear/Q#1P#2IfalseT#3_scaledCostCurve};
\addplot[mark=x, mark size=1, draw=red] table[header=true, col sep=comma] {plotsdata/dwave/Q#1P#2Ifalse_T#3_scaledCostCurve};
\addplot[mark=diamond, mark size=2, draw=black] table[header=true, col sep=comma] {plotsdata/heuristics/Q#1P#2Ifalse_CLIMBT#3_scaledCostCurve};
\addplot[mark=square, mark size=1, draw=brown] table[header=true, col sep=comma] {plotsdata/heuristics/Q#1P#2Ifalse_GEN50T#3_scaledCostCurve};
\addplot[mark=asterisk, mark size=2, draw=orange] table[header=true, col sep=comma] {plotsdata/heuristics/Q#1P#2Ifalse_GEN200T#3_scaledCostCurve};
}

\def\timeCostAllPlots#1#2{
\timeCostPlot{#1}{#2}{0}
\timeCostPlot{#1}{#2}{1}
\timeCostPlot{#1}{#2}{2}
\timeCostPlot{#1}{#2}{3}
\timeCostPlot{#1}{#2}{4}
\timeCostPlot{#1}{#2}{5}
\timeCostPlot{#1}{#2}{6}
\timeCostPlot{#1}{#2}{7}
\timeCostPlot{#1}{#2}{8}
\timeCostPlot{#1}{#2}{9}
\timeCostPlot{#1}{#2}{10}
\timeCostPlot{#1}{#2}{11}
\timeCostPlot{#1}{#2}{12}
\timeCostPlot{#1}{#2}{13}
\timeCostPlot{#1}{#2}{14}
\timeCostPlot{#1}{#2}{15}
\timeCostPlot{#1}{#2}{16}
\timeCostPlot{#1}{#2}{17}
\timeCostPlot{#1}{#2}{18}
\timeCostPlot{#1}{#2}{19}
}

\def\timeCostFigure#1#2{
\begin{figure}[t!]
\centering
\begin{tikzpicture}
\begin{groupplot}[group style={group size=4 by 5, x descriptions at=edge bottom, y descriptions at=edge left, horizontal sep=1pt, vertical sep=1pt}, xmode=log, width=3.375cm,
xlabel=Time (millisec), ylabel=Cost, xlabel near ticks, ylabel near ticks, 
xmajorgrids, 
xtickten={0,1,2,3,4,5}, ticklabel style = {font=\tiny}, xticklabel style={rotate=15}, label style = {font=\tiny}, legend style={font=\tiny}]
\timeCostAllPlots{#1}{#2}
\end{groupplot}
\end{tikzpicture}

\ref{timeCostLegend}
\caption{Solution cost as a function of optimization time for 20 MQO problem instances with #1 queries and #2 plans per query.\label{Q#1P#2Fig}}
\end{figure}
}

\timeCostFigure{537}{2}
\timeCostFigure{253}{3}
\timeCostFigure{140}{4}
\timeCostFigure{108}{5}

We compare optimization approaches in terms of how solution quality, measured by the scaled execution cost of the current plan selection, evolves as a function of optimization time. We measure execution cost at regular time intervals, after 1, 10, 100, 1000, $10^4$, and $10^5$ milliseconds. For the quantum annealer, we report the execution cost of the best solution found after each batch of 10 annealing runs in the following figures (this information is generated by default and using it does not introduce measurement overheads). 

Figures~\ref{Q537P2Fig} to \ref{Q108P5Fig} show the performance results for between two and five alternative plans per query and the associated maximal number of queries that can be treated using the available qubits (between 537 queries for two plans and 108 queries for five plans). Note that the x-axis, on which optimization time is represented, is logarithmic. Each figure shows detailed results for each out of 20 test cases. We chose to represent performance results for single test cases to give an intuition about how consistent the performance differences between the compared approaches are. The figure legends use the abbreviations QA for quantum annealer, LIN-MQO for linear solver applied to MQO problem instances, LIN-QUB for linear solver applied to QUBO instances, CLIMB for iterated hill climbing, and GEN(50) and GEN(200) for the genetic algorithm with population size 50 and 200 respectively. 

\begin{table}[t!]
\centering
\caption{Milliseconds until finding the optimal solution via integer linear programming.\label{optimalSolutionTimeTable}}
\begin{tabular}{llll}
\toprule[1pt]
\textbf{\# Queries} & \textbf{Minimum} & \textbf{Median} & \textbf{Maximum}\\
\midrule[1pt]
537 & 9261 & 25205.5 & 34570 \\
\midrule
253 & 129 & 178.5 & 206 \\
\midrule
140 & 45 & 128 & 241 \\
\midrule
108 & 47 & 48 & 51 \\
\bottomrule[1pt]
\end{tabular}
\end{table}

We first discuss the results shown in Figure~\ref{Q537P2Fig}. The corresponding class of test cases with 537 queries is the hardest class among the ones we consider if judging hardness by the time it takes to find the optimal solution using the linear solver directly on the MQO problem instance (Table~\ref{optimalSolutionTimeTable} shows aggregates of the time it takes to find the optimal solution depending on the number of queries). 

Among the approaches executed on a classical computer, the integer linear programming solver achieves the best results in the optimization time range between 1 and 100 seconds. The performance is clearly better when solving the MQO problem directly instead of the QUBO representation that was derived from it. This is to be expected as the MQO representation leads to a smaller search space than the QUBO representation: only the QUBO representation allows to represent invalid solutions where multiple or no plans are selected for some queries which leads to an exponential blowup in search space size in the number of alternative plans per query. 

The solutions produced by the randomized algorithms are clearly inferior to the ones found by the linear solver after one second of optimization time. Before that time, the hill climbing algorithm often produces slightly better solutions than the linear solver. Over the long term, the hill climbing algorithm is however beaten by the genetic algorithm. This is intuitive as the hill climbing algorithm is the simplest one among all compared approaches. 


All classical approaches have in common that solution quality improves significantly over a time span of several seconds. This is different for the quantum annealer. The execution cost of the solution found after the first annealing run is relatively close to the best execution cost found after 1000 runs with an average cost reduction of 1.5\% from run one to run 1000. As the cost of the final solution after 1000 annealing runs is very close to the optimal solution found by the linear solver (with an average cost overhead of $0.4$\%), this means that the quantum annealer produces good solutions very quickly compared to the other approaches. 

More precisely, in 13 out of 20 test cases, the quantum annealer finds a solution after one annealing run (which takes less than half a millisecond) that is better or equivalent to the solutions found by all other approaches after 10 seconds. The best solution obtained during the first 10 annealing runs is in 18 out of 20 test cases at least equivalent to the solutions generated after 10 seconds by the other approaches. The solution returned by the first annealing run is for all test cases at least equivalent to the solutions generated by the other approaches after one second. This shows that there is a range of MQO problems in which the quantum annealer consistently outperforms the other approaches with a speedup of more than factor 1000.


The performance advantage of the quantum annealer over the other approaches gradually decreases as we increase the number of plans per query which decreases the number of queries that can be represented with the available number of qubits. Figures~\ref{Q253P3Fig} to \ref{Q108P5Fig} show that development. In Figure~\ref{Q108P5Fig}, the quantum annealer is superior in the optimization time range up to 10 milliseconds while the linear solver finds the optimal solution in less than 100 milliseconds. 

We attribute this effect to two related reasons: First, as shown by our analysis in Section~\ref{analysisSec}, increasing the number of alternative plans per query increases the number of qubits required for representing one single logical variable quickly. This means that the search space size of the problems that can be mapped to the quantum annealer decreases. Experimenting with easier problems generally tends to decrease the performance gap between optimization algorithms. On the other side, the ratio between QUBO solutions representing invalid MQO solutions and QUBO solutions representing valid MQO solutions increases exponentially in the number of plans per result. Hence the drawback of having to work with a reformulation of the original problem becomes more significant as the number of plans increases. 

In summary, we have identified a class of MQO problems where the quantum annealer in combination with our mapping method outperforms approaches on classical computers in finding near-optimal solutions by several orders or magnitude. This performance advantage decreases however quickly once the problem instances are less convenient to represent as QUBO problems.






\section{Related Work}
\label{relatedSec}

Our work relates to prior work on MQO, to publications showing how to solve specific problems using a quantum computer, and to experimental evaluations of quantum annealers for specific problem classes. 

The MQO problem~\cite{Sellis1988a} is a classical database-related optimization problem. The goal of MQO is to reduce execution cost by sharing work among queries. This requires preparatory steps such as identifying common expressions among queries and generating alternative plans for each query~\cite{Diego, Jarke1985}. The optimization problem of selecting an optimal combination of plans for execution is orthogonal to the problem of identifying common sub-expressions and generating plans that allow to exploit them. We focus on plan selection. 

Various approaches have been proposed for selecting an optimal combination of plans in MQO. The first generation of MQO approaches were branch-and-bound algorithms or based on the A-* algorithm~\cite{Cosar1993a, Sellis1988a, Shim1994}. Such approaches scale only to a limited number of queries~\cite{Bayir2007} which motivates the use of randomized algorithms such as genetic algorithms~\cite{Bayir2007, Dokeroglu2015} or efficient greedy heuristics such as hill climbing~\cite{Dalvi2003, Dokeroglu2015, Kementsietsidis2008, Mistry2000,  Roy1999}. Approaches based on integer linear programming~\cite{Bellatreche2013, Dokeroglu} have been shown to outperform prior algorithms~\cite{Dokeroglu} if the goal is to find optimal MQO solutions. We selected a representative subset of recently proposed MQO approaches for our experimental evaluation. We did not consider approaches that target specific scenarios (e.g., SPARQL processing~\cite{Le2012}) or approaches that are based on a different problem model than the one we consider (e.g., representations based on And-Or-Dags~\cite{Roy1999}).

Our work connects to other publications showing how to solve specific problems on quantum computers, including for instance database search~\cite{Grover1996}, classification~\cite{Rose2008}, calculation of Ramsey numbers~\cite{Bian2013, Gaitan2012}, some of the classical NP-hard optimization problems~\cite{Lucas2014}, fault detection~\cite{Perdomo-Ortiz2015}, job shop sheduling~\cite{Venturelli2015}, or protein folding~\cite{Perdomo-Ortiz2012a}. Authors affiliated with NASA have recently studied how to solve several optimization problems that are relevant in the context of NASA's future deep  space missions on an adiabatic quantum annealer~\cite{Williams}. None of the aforementioned publications treats however the problem of MQO. Furthermore, not all of the aforementioned publications feature an experimental evaluation. 

One of the first performance evaluations that compares an adiabatic quantum annealer against classical optimizers was published in 2013~by McGeoch et al.~\cite{Mcgeoch2013a}. The evaluated quantum annealer is an earlier version of the one we use in our evaluation; our annealer increases the number of qubits by roughly one order of magnitude compared to the version from 2013 which allows to treat significantly larger search spaces. McGeoch et al.\ compare the quantum annealer against an integer programming solver in terms of the time it takes to find optimal solutions. While the quantum annealer outperforms the integer programming solver by several orders of magnitude, an alternative representation of the integer programming problem has been shown later to decrease that performance gap significantly~\cite{Dash2013}. We use the optimized representation in our experiments. 

A quantum annealer with 512 qubits, the predecessor of the one we experimented with, was recently compared against classical algorithms by multiple groups~\cite{Hen2015, King2015a}. The focus of those evaluations was to compare the asymptotic growth of  optimization time until an optimal solution is found between the quantum annealer and traditional optimization algorithms. Results by Hen et al.~\cite{Hen2015} for a class of Ising problems generated without limiting the weights on and between qubits show slight advantages for the D-Wave annealer only for a very small range of test parameters and no speedup for others while results on weight-limited instances~\cite{King2015a} show a robust scalability advantage for the quantum annealer. 

The focus of our experimental evaluation differs in several ways. First, we focus on the MQO problem and not on Ising problems which are natively supported by the quantum annealer.  This is a challenging scenario for the quantum annealer since the approaches we compare against do not suffer from the blowup in search space size when transforming MQO problems into Ising problems. Second, while the other evaluations essentially compare the quantum annealer against hypothetical massively-parallel classical solvers by scaling down the optimization times of classical solvers, we are interested in the raw optimization times realized by existing systems. Finally, while prior evaluations mostly focus on the time until an optimal solution is found, our evaluation is broader as we consider how solution quality evolves as a function of optimization time.

In addition, our work differs from prior evaluations since we use the newest model of the quantum annealer with over 1000 qubits that was very recently released. To the best of our knowledge, the results in this paper are the first performance results for the D-Wave 2X besides an initial publication by affiliates of the company D-Wave~\cite{King2015}.

\section{Conclusion and Outlook}
\label{conclusionSec}

We have shown how the problem of multiple query optimization can be solved using an adiabatic quantum computer. We analyzed our approach formally and evaluated it experimentally, making this one of the first published experimental evaluations of adiabatic quantum annealers with over 1000 qubits. The quantum annealer finds near-optimal solutions faster than various classical optimization approaches in all evaluated scenarios. The speedup reaches up to three orders of magnitude for a subset of evaluated scenarios. Due to the limited number of qubits, we were only able to compare on a relatively narrow range of problem instances. 

Our current mapping approach transforms one MQO problem instance into one QUBO problem instance. We will explore approaches that map one MQO problem instance into a series of QUBO problems in future work which should in principle allow to treat larger problem instances than the ones we have considered in our experiments. Our experimental results are specific to MQO. We plan to address other database-specific optimization problems in future work.

\section{Acknowledgments}

We acknowledge the support of the Universities Space Research Association Quantum AI Lab Research Opportunity Program. This work was supported by ERC Grant 279804 and by a Google PhD Fellowship. We thank Davide Venturelli for his detailed comments on the paper and Sergio Boixo for his help in using the quantum annealer.

\begin{tiny}
\bibliographystyle{abbrv}

\end{tiny}

\end{document}